\pdfoutput=1
\documentclass[11pt,a4paper]{article}

\usepackage{amsmath,amssymb,amsthm}
\usepackage{braket}
\usepackage{mathtools}
\usepackage{hyperref}
\usepackage{geometry}
\usepackage{enumitem}
\usepackage{microtype}
\usepackage{tikz}
\usetikzlibrary{arrows.meta,positioning,shapes.geometric}
\newtheorem{theorem}{Theorem}[section]
\newtheorem{lemma}[theorem]{Lemma}
\newtheorem{corollary}[theorem]{Corollary}
\newtheorem{proposition}[theorem]{Proposition}
\newtheorem{definition}[theorem]{Definition}
\newtheorem{remark}[theorem]{Remark}
\newtheorem{example}[theorem]{Example}

\DeclareMathOperator{\Tr}{Tr}

\newcommand{\Id}{\mathbb{I}}
\newcommand{\UU}{U_{\mathrm{UQRAM}}}
\newcommand{\normone}[1]{\left\lVert #1 \right\rVert_1}

\title{Read-Only Opacity and Restricted-Access Inference on Quantum Memories via U-QRAM}
\author{Leonardo Bohac}
\date{}

\begin{document}
\maketitle

\begin{abstract}
Universal QRAM (U-QRAM)~\cite{BohacUqram25} is a fixed, data-independent unitary interface that implements coherent random-access reads relative to a designated computational ``truth-table'' basis on the memory register.
This work studies \emph{restricted-access inference}: the memory register is persistent but inaccessible, while an experimenter may prepare and measure only accessible registers and may invoke the fixed read interaction.

Allowing the memory to be in an arbitrary quantum state (pure or mixed, possibly entangled with an inaccessible reference system, or a coherent superposition of truth tables), we establish a sharp, protocol-independent limitation of read-only access.
For \emph{any} finite-query protocol---including arbitrary accessible ancillas, intermediate measurements, adaptivity, and general CPTP processing between queries---the induced output state on the accessible registers depends on the memory state \(\rho_M\) \emph{only} through its diagonal in the truth-table basis.
Equivalently, read-only access factors through dephasing (pinching) in that basis; coherences between distinct truth tables are operationally invisible.

Consequently, every memory-hypothesis testing task reduces to a standard state-discrimination problem on the accessible registers, and the minimum-error optimal measurement is characterized by Helstrom theory.
We illustrate the framework with three explicit examples: (i) the phase-kickback reduction recovering the one-query Bernstein--Vazirani/Deutsch--Jozsa geometry, (ii) a minimal Helstrom instance with optimal success probability \(3/4\), and (iii) perfect indistinguishability of relative phases in entangled truth-table superpositions.
\end{abstract}

\section{Introduction}

A random-access memory is specified less by its physical substrate than by its \emph{interface}: a read operation that returns the contents of a selected cell.
U-QRAM~\cite{BohacUqram25} provides a coherent analogue of this interface.
Crucially, the read interaction is a \emph{fixed and known} unitary, while the memory register itself may be unknown and persistent.
This motivates an information-theoretic question distinct from standard oracle or channel discrimination:
\begin{quote}
Given a fixed U-QRAM device and no direct access to the memory register, what can be inferred about the memory state by probing only through the read interface?
\end{quote}

The present work formalizes this as \emph{restricted-access inference}.
A hypothesis about the stored memory (a state hypothesis) induces, via the fixed interface and a chosen probing protocol, a corresponding hypothesis about the resulting state on the accessible registers.
The latter can be treated using standard quantum-state discrimination theory.

\paragraph{Contributions.}
\begin{itemize}[leftmargin=1.5em]
\item We formalize a general model of read-only restricted-access inference for U-QRAM, allowing arbitrary accessible ancillas, intermediate measurements, adaptivity, and general CPTP processing between read queries.
\item We prove a protocol-independent \emph{read-only opacity} theorem: for any read-only protocol, the induced accessible output depends on the memory state \(\rho_M\) only through its diagonal in the truth-table basis (Theorem~\ref{thm:opacity}).
Equivalently, read-only inference factors through dephasing in that basis, and the induced memory-to-output map is entanglement-breaking (Corollary~\ref{cor:eb}).
\item We show that any binary memory-hypothesis test reduces to Helstrom-optimal discrimination on the induced accessible states (Theorem~\ref{thm:helstrom}), and we give three worked examples.
\end{itemize}

\paragraph{Related work.}
Minimum-error discrimination between quantum states is governed by Helstrom theory~\cite{Helstrom76,Holevo82} (see also~\cite{Barnett09}).
U-QRAM is inspired by the QRAM model of Giovannetti--Lloyd--Maccone~\cite{GLM08}, but emphasizes a data-independent unitary interface as formalized in~\cite{BohacUqram25}.

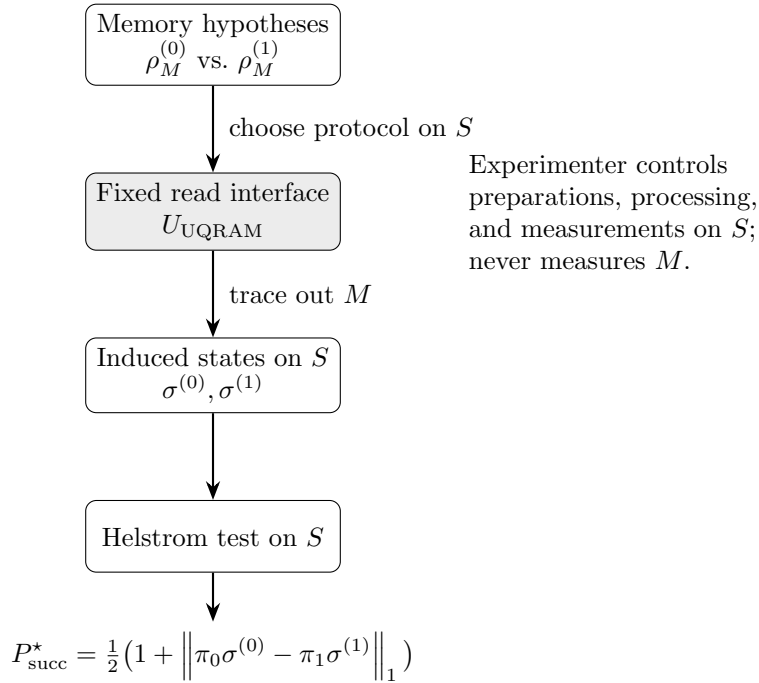
\begin{figure}[ht]
\centering
\begin{tikzpicture}[
    box/.style={draw, rounded corners, minimum width=3.35cm, minimum height=0.95cm, align=center},
    arrow/.style={-{Stealth[length=2.5mm]}, thick},
    every node/.style={font=\small}
]
\node[box] (hyp) {Memory hypotheses\\ \(\rho_M^{(0)}\) vs.\ \(\rho_M^{(1)}\)};
\node[box, below=1.15cm of hyp, fill=gray!15] (uqram) {Fixed read interface\\ \(\UU\)};
\node[box, below=1.15cm of uqram] (induced) {Induced states on \(S\)\\ \(\sigma^{(0)},\sigma^{(1)}\)};
\node[box, below=1.15cm of induced] (helstrom) {Helstrom test on \(S\)};
\node[below=0.65cm of helstrom] (result) {\(\displaystyle P^\star_{\rm succ}=\tfrac12\big(1+\normone{\pi_0\sigma^{(0)}-\pi_1\sigma^{(1)}}\big)\)};

\draw[arrow] (hyp) -- node[right, xshift=2pt] {choose protocol on \(S\)} (uqram);
\draw[arrow] (uqram) -- node[right, xshift=2pt] {trace out \(M\)} (induced);
\draw[arrow] (induced) -- (helstrom);
\draw[arrow] (helstrom) -- (result);

\node[right=1.55cm of uqram, text width=4.0cm, align=left] {Experimenter controls\\ preparations, processing,\\ and measurements on \(S\);\\ never measures \(M\).};
\end{tikzpicture}
\caption{Restricted-access inference: memory hypotheses induce accessible output states via the fixed U-QRAM interface, after which standard state-discrimination theory applies.}
\label{fig:pipeline}
\end{figure}

\section{U-QRAM as a fixed read interface}

\subsection{Registers and the truth-table basis}

We consider three registers:
\begin{itemize}[leftmargin=1.5em]
\item Address register \(A\) of \(n\) qubits, with \(N=2^n\) addresses and computational basis \(\{\ket{a}\}_{a=0}^{N-1}\).
\item Data register \(D\) of one qubit.
\item Memory register \(M\) consisting of \(N\) qubits, one per address.
\end{itemize}

We fix the computational basis \(\{\ket{m}\}_{m\in\{0,1\}^N}\) on \(M\), indexing bit-strings \(m=(m_0,\dots,m_{N-1})\) via
\[
\ket{m}_M := \bigotimes_{j=0}^{N-1}\ket{m_j}.
\]
We emphasize that this is the \emph{global} computational basis on the full \(N\)-qubit memory, so a basis vector \(\ket{m}\) represents a complete table \(m:\{0,\dots,N-1\}\to\{0,1\}\).
This basis is singled out because the read interface is \emph{defined} relative to it.
The actual memory state \(\rho_M\) may be arbitrary (including coherent superpositions across truth tables and/or entanglement with an inaccessible reference).

\subsection{The read unitary}

\begin{definition}[U-QRAM read unitary~\cite{BohacUqram25}]
\label{def:uqram}
The U-QRAM read interaction is the unitary \(\UU\) on \(A\otimes D\otimes M\) defined on computational-basis states by
\begin{equation}
\label{eq:uqram}
\UU \ket{a}_A\ket{y}_D\ket{m}_M
=
\ket{a}_A\ket{y\oplus m_a}_D\ket{m}_M,
\end{equation}
and extended by linearity.
\end{definition}

On basis states \(\ket{m}\), the memory register is left invariant.
For general coherent \(\rho_M\), \(\UU\) may entangle \(M\) with \(A\otimes D\); this interaction is central to the opacity phenomenon studied below.

\subsection{Controlled decomposition and phase kickback}

\begin{lemma}[Controlled decomposition]
\label{lem:controlled}
There exist unitaries \(\{V_m\}_{m\in\{0,1\}^N}\) on \(A\otimes D\) such that
\begin{equation}
\label{eq:controlled}
\UU
=
\sum_{m\in\{0,1\}^N}
\ket{m}\!\bra{m}_M \otimes V_m,
\end{equation}
where \(V_m\) acts as
\(
V_m\ket{a}_A\ket{y}_D=\ket{a}_A\ket{y\oplus m_a}_D.
\)
\end{lemma}

\begin{proof}
Fix \(m\) and define \(V_m\) by the stated action on the basis \(\ket{a}\ket{y}\).
Then~\eqref{eq:controlled} reproduces~\eqref{eq:uqram} on basis states, and extension by linearity yields the claim.
\end{proof}

\begin{lemma}[Phase-kickback reduction]
\label{lem:phase}
Let \(\ket{-}=(\ket{0}-\ket{1})/\sqrt2\).
For any \(m\), on inputs of the form \(\ket{\psi}_A\ket{-}_D\) the unitary \(V_m\) acts as
\[
V_m(\ket{\psi}_A\ket{-}_D)= (O_m\ket{\psi}_A)\ket{-}_D,
\qquad
O_m\ket{a}=(-1)^{m_a}\ket{a}.
\]
\end{lemma}

\begin{proof}
Since \(X\ket{-}=-\ket{-}\), the map \(\ket{y}\mapsto\ket{y\oplus m_a}\) contributes a phase \((-1)^{m_a}\) while leaving \(\ket{-}\) invariant.
\end{proof}

\section{Restricted-access inference and Helstrom optimality}

\subsection{General read-only protocols}

Let \(S\) denote the \emph{accessible system}, consisting of the address and data registers and any additional ancilla/workspace under the experimenter's control:
\[
S := A\otimes D\otimes R,
\]
where \(R\) is arbitrary (finite-dimensional).
The read interaction acts as \(\tilde U:=\UU\otimes \Id_R\) on \(A\otimes D\otimes M\otimes R\).

A general \emph{read-only \(t\)-query protocol} is:
\begin{itemize}[leftmargin=1.5em]
\item Prepare an initial state \(\rho_S\) on \(S\).
\item For \(k=1,\dots,t\): apply \(\tilde U\) (i.e.\ \(\UU\) on \(A\otimes D\otimes M\) and identity on \(R\)), then apply an arbitrary CPTP map \(\mathcal{E}_k\) on \(S\).
This includes intermediate measurements, noise, and adaptivity (by storing classical outcomes in \(R\)).
\item Finally, measure \(S\) with an arbitrary POVM \(\{E_x\}_x\).
\end{itemize}
No operation is applied directly to \(M\), and the memory is never measured.

For a fixed memory state \(\rho_M\), the induced accessible state before the final measurement is
\begin{equation}
\label{eq:induced}
\sigma(\rho_M)
:=
\Tr_M\!\Big[\mathcal{W}\big(\rho_{S}\otimes \rho_M\big)\Big],
\end{equation}
where \(\mathcal{W}\) denotes the overall CPTP evolution on \(S\otimes M\) generated by the interleaving of \(\tilde U\) with \(\{\mathcal{E}_k\}\).

A binary hypothesis test specifies \(\rho_M\in\{\rho_M^{(0)},\rho_M^{(1)}\}\) with priors \(\pi_0,\pi_1\).
The experimenter is then discriminating the induced states
\(
\sigma^{(h)} := \sigma(\rho_M^{(h)}),\ h\in\{0,1\},
\)
using measurements on \(S\) only.

\subsection{Helstrom theory}

\begin{theorem}[Helstrom optimal test~\cite{Helstrom76,Holevo82}]
\label{thm:helstrom}
Let \(\sigma_0,\sigma_1\) be two density operators with priors \(\pi_0,\pi_1\).
Define \(\Delta=\pi_0\sigma_0-\pi_1\sigma_1\).
An optimal two-outcome measurement is given by the projector onto the positive part of \(\Delta\) (decide \(H_0\)) and its complement (decide \(H_1\)).
The optimal success probability is
\[
P_{\mathrm{succ}}^\star = \tfrac12\big(1+\normone{\Delta}\big),
\qquad
\normone{X}:=\Tr\sqrt{X^\dagger X}.
\]
In particular, for equal priors \(\pi_0=\pi_1=\tfrac12\),
\[
P_{\mathrm{succ}}^\star = \tfrac12\Big(1+\tfrac12\normone{\sigma_0-\sigma_1}\Big).
\]
\end{theorem}

\section{Read-only opacity: coherences between truth tables are invisible}

\subsection{Dephasing (pinching) in the truth-table basis}

Let
\begin{equation}
\label{eq:dephase}
\mathcal{D}(\rho_M):=\sum_{m\in\{0,1\}^N} \ket{m}\!\bra{m}\,\rho_M\,\ket{m}\!\bra{m}
\end{equation}
denote dephasing in the truth-table basis.
Write \(p(m)=\bra{m}\rho_M\ket{m}\) for the resulting classical distribution over truth tables.

\begin{lemma}[Commutation with dephasing]
\label{lem:commute}
Let \(\mathcal{D}\) be the dephasing map on \(M\) defined in~\eqref{eq:dephase}, and extend it to \(S\otimes M\) as \(\Id_S\otimes \mathcal{D}\).
Then:
\begin{enumerate}[leftmargin=1.5em]
\item Conjugation by \(\tilde U\) commutes with \(\Id_S\otimes\mathcal{D}\), i.e.
\[
(\Id_S\otimes\mathcal{D})\!\big(\tilde U\, X \,\tilde U^\dagger\big)
=
\tilde U\big((\Id_S\otimes\mathcal{D})(X)\big)\tilde U^\dagger
\quad \text{for all operators } X \text{ on } S\otimes M.
\]
\item Any CPTP map \(\mathcal{E}\) acting only on \(S\) commutes with \(\Id_S\otimes\mathcal{D}\), i.e.
\[
(\Id_S\otimes\mathcal{D})\!\big((\mathcal{E}\otimes \Id_M)(X)\big)
=
(\mathcal{E}\otimes \Id_M)\!\big((\Id_S\otimes\mathcal{D})(X)\big).
\]
\end{enumerate}
\end{lemma}

\begin{proof}
For (1), use Lemma~\ref{lem:controlled}: \(\UU=\sum_m \ket{m}\!\bra{m}\otimes V_m\), hence \(\UU\) is block-diagonal in the truth-table basis on \(M\), and so is \(\tilde U=\UU\otimes \Id_R\).
In particular, \(\tilde U\) commutes with each projector \(\Pi_m:=\ket{m}\!\bra{m}\) on \(M\), i.e.\ \((\Id_S\otimes \Pi_m)\tilde U=\tilde U(\Id_S\otimes \Pi_m)\).
Since the pinching map is \(\mathcal{D}(X)=\sum_m \Pi_m X \Pi_m\), it follows that \(\Id_S\otimes \mathcal{D}\) commutes with conjugation by \(\tilde U\).

For (2), \(\mathcal{E}\otimes \Id_M\) acts trivially on \(M\), so it commutes with any map acting only on \(M\), in particular \(\Id_S\otimes\mathcal{D}\).
\end{proof}

\subsection{Opacity theorem}

\begin{theorem}[Read-only opacity / dephasing reduction]
\label{thm:opacity}
Fix any read-only \(t\)-query protocol in the sense of Section~3.1 (allowing arbitrary accessible ancillas, intermediate measurements and adaptivity, and general CPTP processing between queries).
Then the induced accessible state depends on \(\rho_M\) only through its diagonal in the truth-table basis:
\[
\sigma(\rho_M)=\sigma(\mathcal{D}(\rho_M)).
\]
Equivalently, there exist states \(\{\sigma_m\}_{m\in\{0,1\}^N}\) on the accessible system \(S\), determined by the protocol (but not by \(\rho_M\)), such that
\begin{equation}
\label{eq:mixture}
\sigma(\rho_M)=\sum_{m\in\{0,1\}^N} p(m)\,\sigma_m,
\end{equation}
where \(\sigma_m\) is exactly the induced output obtained if the memory were initialized to \(\ket{m}\bra{m}\).
Consequently, no read-only protocol can distinguish two memory states that share the same diagonal \(p(m)\) in the truth-table basis.
\end{theorem}

\begin{proof}
Let \(\mathcal{W}\) denote the overall CPTP map on \(S\otimes M\) induced by the protocol.
By Lemma~\ref{lem:commute}, each use of \(\tilde U\) and each interleaving CPTP map on \(S\) commutes with \(\Id_S\otimes\mathcal{D}\); hence so does their composition \(\mathcal{W}\):
\[
(\Id_S\otimes\mathcal{D})\circ \mathcal{W}
=
\mathcal{W}\circ(\Id_S\otimes\mathcal{D}).
\]
Now use the fact that tracing out \(M\) is invariant under dephasing on \(M\), i.e.
\(\Tr_M=\Tr_M\circ(\Id_S\otimes\mathcal{D})\).
Then
\begin{align*}
\sigma(\rho_M)
&= \Tr_M\!\big[\mathcal{W}(\rho_S\otimes\rho_M)\big] \\
&= \Tr_M\!\Big[(\Id_S\otimes\mathcal{D})\big(\mathcal{W}(\rho_S\otimes\rho_M)\big)\Big] \\
&= \Tr_M\!\Big[\mathcal{W}\big((\Id_S\otimes\mathcal{D})(\rho_S\otimes\rho_M)\big)\Big] \\
&= \Tr_M\!\big[\mathcal{W}(\rho_S\otimes\mathcal{D}(\rho_M))\big]
= \sigma(\mathcal{D}(\rho_M)).
\end{align*}
For the mixture form, write \(\mathcal{D}(\rho_M)=\sum_m p(m)\ket{m}\!\bra{m}\) and use linearity:
\[
\sigma(\rho_M)=\sum_m p(m)\,\Tr_M\!\big[\mathcal{W}(\rho_S\otimes\ket{m}\!\bra{m})\big].
\]
Define \(\sigma_m:=\Tr_M[\mathcal{W}(\rho_S\otimes\ket{m}\!\bra{m})]\) to obtain~\eqref{eq:mixture}.
\end{proof}

\begin{remark}[Interpretation]
Theorem~\ref{thm:opacity} does \emph{not} assert that the memory is classical.
It asserts that \emph{read-only access} factors through dephasing in the truth-table basis: coherences between distinct truth tables cannot influence any statistics obtainable without acting on or measuring \(M\).
\end{remark}

\begin{corollary}[Measure-and-prepare form; entanglement-breaking]
\label{cor:eb}
For any fixed read-only protocol, the induced map \(\Phi:\rho_M\mapsto \sigma(\rho_M)\) is a classical-to-quantum (measure-and-prepare) channel:
\[
\Phi(\rho_M)=\sum_m \Tr(\ket{m}\!\bra{m}\rho_M)\,\sigma_m.
\]
In particular, \(\Phi\) is entanglement-breaking (equivalently, measure-and-prepare)~\cite{HSR03}: for any joint state \(\rho_{MQ}\) of \(M\) with an arbitrary reference system \(Q\), the state \((\Phi\otimes \Id_Q)(\rho_{MQ})\) is separable across \(S{:}Q\).
\end{corollary}

\begin{remark}[Extension to memories entangled with a reference]
The dephasing reduction holds at the level of joint outputs.
If the initial memory is part of an arbitrary joint state \(\rho_{MQ}\), then the induced state on \(S\otimes Q\) satisfies
\[
(\Phi\otimes \Id_Q)(\rho_{MQ})
=
(\Phi\otimes \Id_Q)\big((\mathcal{D}\otimes \Id_Q)(\rho_{MQ})\big),
\]
i.e.\ only the truth-table diagonal blocks \(\{(\Pi_m\otimes\Id_Q)\rho_{MQ}(\Pi_m\otimes\Id_Q)\}_m\) can influence any read-only observable on \(S\) (even in the presence of side information \(Q\)).
\end{remark}

\begin{remark}[Reduction to a classical experiment; a universal bound]
Fix a read-only protocol, and let \(p_h(m)=\bra{m}\rho_M^{(h)}\ket{m}\) be the truth-table distributions under hypotheses \(H_h\).
By Theorem~\ref{thm:opacity}, the induced states satisfy
\(
\sigma^{(h)}=\sum_m p_h(m)\,\sigma_m
\)
for a protocol-dependent family \(\{\sigma_m\}_m\).
Thus, for a fixed protocol, memory discrimination is equivalent to distinguishing two classical distributions \(p_0,p_1\) passed through the classical-to-quantum channel \(m\mapsto\sigma_m\).

Define the total-variation distance
\[
\mathrm{TV}(p_0,p_1):=\frac12\sum_m |p_0(m)-p_1(m)|=\frac12\|p_0-p_1\|_1.
\]
Then, for any such protocol,
\[
\frac12\big\lVert \sigma^{(0)}-\sigma^{(1)}\big\rVert_1
=
\frac12\Big\lVert \sum_m \big(p_0(m)-p_1(m)\big)\sigma_m\Big\rVert_1
\le \frac12\sum_m \big|p_0(m)-p_1(m)\big|\cdot \|\sigma_m\|_1
= \mathrm{TV}(p_0,p_1),
\]
where \(\|\sigma_m\|_1=\Tr(\sigma_m)=1\) for each density operator \(\sigma_m\).
In particular, for equal priors the optimal success probability obeys
\(
P_{\mathrm{succ}}^\star \le \tfrac12\big(1+\mathrm{TV}(p_0,p_1)\big).
\)
\end{remark}

\begin{proposition}[A diagonal-world quantitative bound: tightness and an output-dimension obstruction]
\label{prop:quant}
Fix a read-only protocol, with induced cq-channel \(m\mapsto \sigma_m\) as in~\eqref{eq:mixture}, and let \(\sigma^{(h)}=\sum_m p_h(m)\sigma_m\) be the induced states under two memory-diagonal hypotheses \(p_0,p_1\).

\begin{enumerate}[leftmargin=1.5em]
\item (Tightness of the TV bound.) Let \(\delta(m):=p_0(m)-p_1(m)\) and assume \(\mathrm{TV}(p_0,p_1)>0\).
Define
\[
\alpha:=\mathrm{TV}(p_0,p_1)=\sum_{\delta(m)>0}\delta(m),
\quad
q_+(m):=\frac{\max\{\delta(m),0\}}{\alpha},
\quad
q_-(m):=\frac{\max\{-\delta(m),0\}}{\alpha},
\]
and the corresponding states
\[
\tau_+ := \sum_m q_+(m)\,\sigma_m,
\qquad
\tau_- := \sum_m q_-(m)\,\sigma_m.
\]
Then
\[
\frac12\|\sigma^{(0)}-\sigma^{(1)}\|_1
=
\alpha\cdot \frac12\|\tau_+-\tau_-\|_1
\le \alpha=\mathrm{TV}(p_0,p_1).
\]
Moreover, equality (i.e.\ saturation of the TV bound) holds if and only if \(\tau_+\) and \(\tau_-\) are perfectly distinguishable, equivalently their supports are orthogonal.

\item (Output-dimension obstruction for multi-hypothesis identification.) Let the memory be promised to be one of \(K\) deterministic truth tables \(\{m_1,\dots,m_K\}\), each with equal prior, and let \(d:=\dim(S)\).
Then for any read-only protocol and any measurement on \(S\), the optimal success probability for identifying \(m_i\) from the induced states \(\{\sigma_{m_i}\}_{i=1}^K\) satisfies
\[
P^\star_{\mathrm{succ}} \le \frac{d}{K}.
\]
In particular, perfect identification (\(P^\star_{\mathrm{succ}}=1\)) requires \(d\ge K\).
\end{enumerate}
\end{proposition}

\begin{remark}[Operational consequence]
Corollary~\ref{cor:eb} formalizes a strong ``no-coherence leakage'' statement: regardless of probe design, adaptivity, or number of read queries, read-only interaction cannot transmit coherence between distinct truth tables into accessible degrees of freedom.
Proposition~\ref{prop:quant} further shows that (i) the universal TV bound is saturated exactly when the protocol collapses the task to a \emph{perfect} binary test between two diagonal mixtures, and (ii) finite accessible output dimension imposes a sharp obstruction to identifying large families of candidate tables.
\end{remark}

\section{Worked examples}

For concreteness, Examples~\ref{ex:phasekick}--\ref{ex:phaseblind} use the smallest nontrivial random-access setting \(n=1\) (so \(N=2\) addresses and \(M\) consists of two qubits).
In these examples we take the ancilla \(R\) to be trivial, so \(\tilde U=\UU\).
This keeps formulas explicit, while Theorem~\ref{thm:opacity} holds for arbitrary \(n\), arbitrary ancillas, and arbitrary read-only protocols.

\subsection{Example I: phase-kickback and one-query identification geometry}

\begin{example}[Two truth tables via phase kickback]
\label{ex:phasekick}
Consider two deterministic memory tables \(m\in\{00,01\}\).
Prepare the probe \(\ket{+}_A\ket{-}_D\), apply \(\UU\) once, and then measure \(A\) in the \(X\)-basis.
By Lemma~\ref{lem:phase}, the data register factors out and the address acquires phases \((-1)^{m_a}\):
\[
m=00 \ \Rightarrow\  O_m\ket{+}=\ket{+},
\qquad
m=01 \ \Rightarrow\  O_m\ket{+}=\ket{-}.
\]
Thus the two hypotheses are perfectly distinguishable with one query.
This is the familiar phase-oracle geometry used, in particular, in Bernstein--Vazirani and Deutsch--Jozsa type settings~\cite{BV97,DJ92}.
\end{example}

\subsection{Example II: a minimal Helstrom instance with \(P^\star_{\rm succ}=3/4\)}

\begin{example}[Overlapping priors and the Helstrom optimum]
\label{ex:helstrom34}
Let Nature choose a basis table according to one of two \emph{memory-diagonal} ensembles:
\[
H_0:\ m\in\{00,11\}\ \text{uniformly},
\qquad
H_1:\ m\in\{00,01\}\ \text{uniformly},
\]
with equal hypothesis priors \(\pi_0=\pi_1=\tfrac12\).
Use the same one-query protocol as in Example~\ref{ex:phasekick}.

Under \(H_0\), the induced state on \(A\) is \(\ket{+}\bra{+}\): both \(m=00\) and \(m=11\) act as the same phase oracle on \(A\) up to a global phase.
Under \(H_1\), the induced state on \(A\) is the uniform mixture of \(\ket{+}\) and \(\ket{-}\), i.e.\ \(\Id_A/2\).
Therefore the induced accessible states are
\[
\sigma^{(0)}=\ket{+}\!\bra{+}_A\otimes \ket{-}\!\bra{-}_D,
\qquad
\sigma^{(1)}=\tfrac{\Id_A}{2}\otimes \ket{-}\!\bra{-}_D.
\]
Since the \(D\) factor is identical, the trace distance is determined by the \(A\) marginal.
On \(A\), \(\ket{+}\!\bra{+}-\Id/2\) has eigenvalues \(\pm\tfrac12\), hence trace norm \(1\).
Thus \(\normone{\sigma^{(0)}-\sigma^{(1)}}=1\), and Theorem~\ref{thm:helstrom} yields
\[
P_{\mathrm{succ}}^\star=\tfrac12\Big(1+\tfrac12\cdot 1\Big)=\tfrac34.
\]
An optimal decision rule is: measure \(A\) in the \(X\)-basis and decide \(H_0\) iff the outcome is \(+\).
\end{example}

\subsection{Example III: indistinguishability of relative phases in entangled memories}

\begin{example}[A phase that read-only access cannot detect]
\label{ex:phaseblind}
Assume the memory is promised to be in one of two entangled states
\[
H_0:\ket{\Phi^+}_M=\tfrac{1}{\sqrt2}(\ket{00}+\ket{11}),
\qquad
H_1:\ket{\Phi^-}_M=\tfrac{1}{\sqrt2}(\ket{00}-\ket{11}),
\]
with equal priors.
These states have identical diagonals in the truth-table basis: probability \(1/2\) on \(00\) and probability \(1/2\) on \(11\).
By Theorem~\ref{thm:opacity}, \emph{every} read-only protocol induces the same accessible output state under \(H_0\) and \(H_1\).
Therefore
\[
\sigma^{(0)}=\sigma^{(1)},
\qquad
P_{\mathrm{succ}}^\star=\tfrac12.
\]
\end{example}

Example~\ref{ex:phaseblind} isolates the operational boundary captured by Theorem~\ref{thm:opacity}: coherent addressing can generate nontrivial quantum states on accessible registers, but relative phases \emph{between} distinct truth tables stored in memory are not observable under read-only access.

\section{Discussion and outlook}

U-QRAM provides a clean operational setting where the interface is fixed and the uncertainty resides in a persistent (possibly quantum) memory.
In this setting, hypotheses about the memory induce hypotheses about accessible output states, making minimum-error inference a direct application of standard state-discrimination theory.

The main structural result, Theorem~\ref{thm:opacity}, shows that read-only access imposes a protocol-independent constraint: for any finite-query strategy (including ancillas, adaptivity, and arbitrary CPTP processing), the induced accessible state depends on \(\rho_M\) only through its diagonal in the truth-table basis.
Equivalently, the induced map from memory to accessible outputs is a measure-and-prepare channel (Corollary~\ref{cor:eb}), so it cannot transmit coherence or entanglement between distinct truth tables into accessible degrees of freedom.
Beyond the reduction itself, Proposition~\ref{prop:quant} provides a quantitative perspective: the universal TV bound is saturated exactly in the presence of a perfect binary separation induced by the protocol, while finite output dimension obstructs identification of large candidate families of truth tables.

Two directions appear particularly natural.
First, for restricted hypothesis families (e.g.\ memory states supported on a structured set of truth tables), one can study optimal query complexity and optimal probe/measurement design by maximizing the Helstrom distinguishability of the induced states.
In some special cases, the resulting induced-state optimization can be carried out in fully explicit circuit form; see, e.g.,~\cite{BohacClosedForm25} for a closed-form single-query identification instance.
Second, one may formalize \emph{extended} memory interfaces that go beyond read-only semantics---for example, by allowing controlled basis-mixing on \(M\) or limited write operations---and quantify which additional features of \(\rho_M\) (such as coherences or entanglement witnesses) become observable.

\appendix

\section{Proof of Proposition~\ref{prop:quant}}
\label{app:quant}

\paragraph{Part (1): tightness of the TV bound.}
Write \(\delta(m)=p_0(m)-p_1(m)\) and \(\alpha=\mathrm{TV}(p_0,p_1)=\sum_{\delta>0}\delta(m)\).
By construction, \(q_+,q_-\) are distributions supported on \(\{\delta>0\}\) and \(\{\delta<0\}\), and
\[
\delta(m)=\alpha\,q_+(m)-\alpha\,q_-(m).
\]
Therefore
\[
\sigma^{(0)}-\sigma^{(1)}
=\sum_m \delta(m)\sigma_m
=\alpha\Big(\sum_m q_+(m)\sigma_m - \sum_m q_-(m)\sigma_m\Big)
=\alpha(\tau_+-\tau_-),
\]
so
\[
\frac12\|\sigma^{(0)}-\sigma^{(1)}\|_1
=\alpha\cdot \frac12\|\tau_+-\tau_-\|_1
\le \alpha,
\]
since \(\tfrac12\|\tau_+-\tau_-\|_1\le 1\) for any pair of states.
Moreover, \(\tfrac12\|\tau_+-\tau_-\|_1=1\) holds if and only if \(\tau_+\) and \(\tau_-\) are perfectly distinguishable, equivalently their supports are orthogonal.

\paragraph{Part (2): output-dimension obstruction.}
Let \(\{\sigma_{m_i}\}_{i=1}^K\) be the induced states on a \(d\)-dimensional system \(S\), with equal priors.
For any POVM \(\{E_i\}_{i=1}^K\) on \(S\),
\[
P_{\mathrm{succ}}
=\frac{1}{K}\sum_{i=1}^K \Tr(E_i\sigma_{m_i})
\le \frac{1}{K}\sum_{i=1}^K \Tr(E_i),
\]
since \(0\le \sigma_{m_i}\le \Id_S\) implies \(\Tr(E_i\sigma_{m_i})\le \Tr(E_i)\).
Using \(\sum_{i=1}^K E_i=\Id_S\), we get \(\sum_i \Tr(E_i)=\Tr(\Id_S)=d\), hence
\[
P_{\mathrm{succ}} \le \frac{d}{K}.
\]
Maximizing over POVMs yields \(P^\star_{\mathrm{succ}}\le d/K\), and perfect identification requires \(d\ge K\).


\end{document}